\documentclass[reqno,a4paper]{amsart}
\usepackage[linktocpage]{hyperref} 
\usepackage{amssymb}
\usepackage{mathrsfs}
\usepackage{eucal}
\usepackage{ipa} 
\usepackage{color}
\usepackage{graphicx}
\usepackage{comment}

\numberwithin{equation}{section}

\allowdisplaybreaks[2]

\theoremstyle{plain}
\newtheorem{theorem}{Theorem}[section]

\newtheorem{lemma}[theorem]{Lemma}

\newtheorem{definition}[theorem]{Definition}
\newtheorem{remark}[theorem]{Remark}

\newcommand{\Co}{{\mathbb C}}         
\newcommand{\eps}{\epsilon}

\newcommand{\Mcal}{\mathcal M}

\newcommand{\sDiv}{\mathscr{D}}
\newcommand{\sCurl}{\mathscr{C}}
\newcommand{\sCurlDagger}{\mathscr{C}^\dagger}
\newcommand{\sTwist}{\mathscr{T}}

\newcommand{\SL}{\mathrm{SL}}

\newcommand{\met}{g}

\newcommand{\GenVec}{\nu} 

\title[A new tensorial conservation law for Maxwell fields]{A new tensorial conservation law for Maxwell fields on the Kerr background}

\author[L. Andersson]{Lars Andersson} \email{laan@aei.mpg.de}
\address{Albert Einstein Institute, Am M\"uhlenberg 1, D-14476 Potsdam,
  Germany 
\and
Department of Mathematics, Royal Institute of Technology, SE-100 44 Stockholm, Sweden
}

\author[T. B\"ackdahl]{Thomas B\"ackdahl} \email{t.backdahl@ed.ac.uk}
\address{The School of Mathematics, University of Edinburgh, James Clerk Maxwell Building, 
Peter Guthrie Tait Road, Edinburgh
EH9 3FD, UK}

\author[P. Blue]{Pieter Blue} \email{P.Blue@ed.ac.uk}
\address{The School of Mathematics and the Maxwell Institute, University of Edinburgh, James Clerk Maxwell Building, 
Peter Guthrie Tait Road, Edinburgh
EH9 3FD,UK}

\begin{document} 

 \begin{abstract}
A new, conserved, symmetric tensor field for a source-free Maxwell test field on a four-dimensional spacetime with a conformal Killing-Yano tensor, satisfying a certain compatibility condition, is introduced.  In particular, this construction works for the Kerr spacetime.
\end{abstract}

\date{December 9, 2014}

\maketitle

\section{Introduction} 
\label{sec:intro}
In this paper, we consider the Maxwell equation for a real $2$-form
$F_{ab}=F_{[ab]}$,
\begin{align}\label{eq:MaxwellTensor}
\nabla^a F_{ab} &= 0, \quad \nabla^a {*} F_{ab} = 0,
 \end{align}
on a four-dimensional Lorentzian manifold $(\Mcal,\met_{ab})$. Recall
that a
conformal Killing-Yano tensor is real a $2$-form $Y_{ab} = Y_{[ab]}$
satisfying
\begin{align}
\label{eq:CKYdefTensorVersion}
\nabla_{(a}Y_{b)c}={}&- \tfrac{1}{3} g_{ab} \nabla_{d}Y_{c}{}^{d}
 + \tfrac{1}{3} g_{(a|c|}\nabla^{d}Y_{b)d}.
\end{align}
Associated with $Y_{ab}$ is the complex $1$-form 
\begin{align}
\xi_{a}=
{}&
\tfrac{1}{3}i \nabla_{b}Y_{a}{}^{b}
 -  \tfrac{1}{3} \nabla_{b}{*}Y_{a}{}^{b}.
\label{eq:xidefTensorVersion}
\end{align}
We say that $Y_{ab}$ satisfies the aligned matter condition if the
Ricci curvature and $Y_{ab}$ satisfy
\begin{align}
R_{(a}{}^{c}Y_{b)c}={}&0, \quad
R_{(a}{}^{c}{*}Y_{b)c}=
0. 
\label{eq:alignedmatterTensorVersion}
\end{align}

\begin{theorem}
\label{thm:MainResultTensorVersion}
Let $Y_{ab}$ and $F_{ab}$ be real $2$-forms. 
Define the real $2$-form $Z_{ab}$ and the complex
$1$-form $\eta_a$ by
\begin{align}
Z_{ab}={}&- \tfrac{4}{3} ({*}F)_{[a}{}^{c}Y_{b]c},\label{eq:thetadefTensorVersion}\\
\eta_{a}={}&- \tfrac{1}{2} \nabla_{b}Z_{a}{}^{b}
 -  \tfrac{1}{2}i \nabla_{b}{*}Z_{a}{}^{b},
\label{eq:etadefTensorVersion}
\end{align}
and the
real symmetric $2$-tensor $V_{ab}$ by 
\begin{align}
V_{ab}={}&\eta_{(a}\bar{\eta}_{b)}
- \tfrac{1}{2} g_{ab} \eta^{c} \bar{\eta}_{c}
-  \tfrac{1}{3} (\mathcal{L}_{Re\xi}F)_{(a}{}^{c}Z_{b)c}
 + \tfrac{1}{12} g_{ab} (\mathcal{L}_{Re\xi}F)^{cd} Z_{cd}\nonumber\\
& + \tfrac{1}{3} (\mathcal{L}_{Im\xi}{*}F)_{(a}{}^{c}Z_{b)c}
 -  \tfrac{1}{12} g_{ab} (\mathcal{L}_{Im\xi}{*}F)^{cd}
 Z_{cd},\label{eq:VdefTensorVersion}
\end{align}
where $\xi_a$ is given by equation
\eqref{eq:xidefTensorVersion} and $\bar{\eta}_a$ denotes the complex
conjugate of $\eta_a$.

If $Y_{ab}$ is a conformal Killing-Yano tensor satisfying the aligned matter condition
\eqref{eq:alignedmatterTensorVersion} and $F_{ab}$ satisfies the
Maxwell equations \eqref{eq:MaxwellTensor}, 
then $V_{ab}$ has vanishing divergence, $\nabla^a V_{ab} = 0$. 
\end{theorem}
\begin{remark}   
\begin{enumerate} 
\item The vector field $\xi^a$ is Killing, $\nabla_{(a} \xi_{b)} = 0$,  if the aligned matter condition \eqref{eq:alignedmatterTensorVersion} holds, cf. equation \eqref{eq:Twistxi} below. If $\nabla^a Y_{ab} = 0$ then $Y_{ab}$ is a Killing-Yano tensor, in which case $\xi_a$ is real, and the last two terms of \eqref{eq:VdefTensorVersion} vanish.
\item  The Kerr family of stationary, rotating vacuum black hole metrics admit a Killing-Yano tensor. More generally, the Kerr-Newman family of stationary, rotating electro-vacuum black hole metrics admit a Killing-Yano tensor satisfying the aligned matter condition. 
See section \ref{sec:kerr-remarks}  for further discussion. 
\end{enumerate}
\end{remark} 

Let 
\begin{align*}
T_{ab}&= -F_{a}{}^{c}F_{bc} +\tfrac14\met_{ab}F_{cd}F^{cd}
\end{align*}
be the symmetric energy-momentum tensor for the Maxwell field. It is traceless and satisfies the dominant energy condition, i.e. $T_{ab} \mu^a\GenVec^b\geq 0$ for any future causal vectors $\mu^a$, $\GenVec^b$. Further, if $F_{ab}$ satisfies the Maxwell equations, $T_{ab}$ is conserved, $\nabla^a T_{ab} = 0$. 
Hence, the current 
\begin{equation}\label{eq:stresscurr}
J_a = T_{ab} \GenVec^b
\end{equation}  
is conserved, $\nabla^a J_a = 0$, if $\GenVec^a$ is a conformal Killing field, $\nabla_{(a} \GenVec_{b)} - \frac{1}{4} \nabla_c \GenVec^c \met_{ab} = 0$. 

For the Maxwell field on Minkowski space, and more generally on spacetimes admitting conformal Killing-Yano tensors satisfying the aligned matter condition, there are non-classical conserved currents not equivalent\footnote{A conserved current $J_a$ is a 1-form concomitant of the Maxwell field, satisfying $\nabla^a J_a = 0$. We say that $J_a$ is equivalent to $\tilde{J}_a$ if $J_a - \tilde J_a = \nabla^b C_{ab}$ for some 2-form $C_{ab} = C_{[ab]}$.} 
to any of the classical conserved energy-momentum currents of the form \eqref{eq:stresscurr}, see \cite{anco:pohjanpelto:2003:ProcRSoc:MR1997098} and references therein. 
For the Maxwell field on Minkowski space, 
these include  chiral currents constructed using the 20-dimensional family of conformal Killing-Yano tensors of Minkowski space. As shown by the authors \cite{ABB:currents}, analogous conserved currents exist also on spacetimes with conformal Killing-Yano tensors satisfying the aligned matter condition. 

In spite of the large literature on conformal Killing-Yano tensors, and the related conservation laws, the tensorial conservation law exhibited in Theorem \ref{thm:MainResultTensorVersion} appears to be new, even in the Minkowski case. 
The fact that the new  higher order tensor concomitant $V_{ab}$ is conserved also in the case of the Kerr and Kerr-Newman spacetimes makes it interesting from the point of view of the black hole stability problem, which in fact served as an important motivation for the investigation which led to its discovery. See section \ref{sec:kerr-remarks} below for further remarks.  

At this point, we should mention that the symmetric tensor  
$$
B_{ab} =\nabla_{d}F_{bc} \nabla^{d}F_{a}{}^{c} -  \tfrac{1}{4} g_{ab} \nabla_{f}F_{cd} \nabla^{f}F^{cd}
$$
which arises as a trace of the 4-index Chevreton tensor, was shown by Bergqvist et
al. \cite{bergqvist:etal:2003CQGra..20.2663B} to be traceless and conserved for a Maxwell field on a Ricci flat spacetime. Like the conserved tensor $V_{ab}$
introduced in this paper, the tensor $B_{ab}$ introduced by Bergqvist et al. depends on the
$F_{ab}$ and its first derivatives. However, while $B_{ab}$ is traceless and fails to satisfy any positivity
condition, the new tensor $V_{ab}$ has trace $V^{a}{}_{a} = - \eta^{a} \bar{\eta}_{a}$ 
and satisfies a weak form of the dominant energy condition in the sense the its leading order term,  $\eta_{(a}\bar{\eta}_{b)}- \tfrac{1}{2} g_{ab} \eta^{c} \bar{\eta}_{c}$ which is quadratic in first derivatives of $F_{ab}$, is a superenergy tensor for $\eta_a$ and hence does satisfy the dominant energy condition. Hence, energies can be constructed in terms of $V_{ab}$ which are non-negative up to terms of lower order.

The proof of Theorem \ref{thm:MainResultTensorVersion}, which will be given in the next section, makes use of computations in the 2-spinor formalism. In the investigations leading to the main result, the \emph{SymManipulator} package \cite{Bae11a}, developed by one of the authors (T.B.) for the Mathematica based symbolic differential geometry suite \emph{xAct} \cite{xAct}, has played an essential role. SymManipulator makes it possible to systematically exploit decompositions in terms of irreducible representations of the spin group $\SL(2,\Co)$, and allows one to carry out investigations that are not feasible by hand. 

In section \ref{sec:kerr-remarks}, we show how the main result applies for the Kerr-Newman family of electro-vacuum spacetimes, and indicate its relation to the Teukolsky and Teukolsky-Starobinsky equations.

\section{Proof of theorem \ref{thm:MainResultTensorVersion}}
\label{sec:VInSpinors}
For the remainder of this paper, we will make use of the 2-spinor
formalism, following the conventions of \cite{Penrose:1986fk}. 
Since our considerations are local, we can assume without loss of generality that $(\Mcal, \met_{ab})$ is oriented and globally hyperbolic. This also implies that $\Mcal$ is spin.

The spin group is $\SL(2,\Co)$ which has the inequivalent spinor representations $\Co^2$ and $\bar{\Co}^2$.  Unprimed upper case latin indices and their primed versions are used for sections of the corresponding spinor bundles, respectively. The correspondence between spinors and tensors makes it possible to translate all tensor expressions to spinor form. The action of $\SL(2,\Co)$ on $\Co^2$ leaves invariant the spin metric $\eps_{AB} = \eps_{[AB]}$, which is used to raise and lower indices on tensors.  The metric $\met_{ab}$ is related to $\eps_{AB}$ by $\met_{ab} = \eps_{AB} \bar\eps_{A'B'}$.  
Let $\mathcal{S}_{k,l}$ denote the
space of symmetric spinors with $k$ unprimed indices and $l$
primed indices.

There are symmetric spinors $\kappa_{AB}$, $\phi_{AB}$, and $\Theta_{AB}$ such that  
\begin{align*}
Y_{ab}={}&\tfrac{3}{2}i (\bar\epsilon_{A'B'} \kappa_{AB}
 -  \epsilon_{AB} \bar{\kappa}_{A'B'}) ,\\
F_{ab}={}& \bar\epsilon_{A'B'} \phi_{AB}
 + \epsilon_{AB} \bar{\phi}_{A'B'},\\
Z_{ab}={}&\bar\epsilon_{A'B'} \Theta_{AB}
 + \epsilon_{AB} \overline{\Theta}_{A'B'} .
\end{align*}
The normalization of $Y_{ab}$ is chosen for convenience.
Equations \eqref{eq:MaxwellTensor}-\eqref{eq:VdefTensorVersion}
become respectively
\begin{align}
\nabla^{A}{}_{A'}\phi_{AB} &= 0,
\label{eq:Maxwell}\\
\nabla_{(A|A'|}\kappa_{BC)}&=0, 
\label{eq:KillingSpinordef}\\
\xi_{AA'} &=\nabla^{B}{}_{A'}\kappa_{AB},
\label{eq:xidef}\\
\Phi_{(A}{}^{C}{}_{|A'B'|}\kappa_{B)C}&=0.
\label{eq:alignedmatter}\\
\Theta_{AB}&=-2 \kappa_{(A}{}^{C}\phi_{B)C},
\label{eq:thetadef}\\
\eta_{AA'}&=\nabla^{B}{}_{A'}\Theta_{AB}. 
\label{eq:etadef}\\
\intertext{and}
V_{ABA'B'}=\tfrac{1}{2} \eta_{AB'} \bar{\eta}_{A'B}
 + \tfrac{1}{2} \eta_{BA'} \bar{\eta}_{B'A}
 &+ \tfrac{1}{3} \Theta_{AB} (\hat{\mathcal{L}}_{\bar\xi}\bar{\phi})_{A'B'}
 + \tfrac{1}{3} \bar\Theta_{A'B'}
 (\hat{\mathcal{L}}_{\xi}\phi)_{AB}, 
\label{eq:Vdef}
\end{align}
where $\hat{\mathcal{L}}_{\xi}$ is a conformally weighted Lie derivative on spinors, see equation \eqref{eq:Liespin-def} below. 

The projection of the spinor covariant derivative $\nabla_{AA'}$ on symmetric spinors (which form the irreducible representations of the spin group $\SL(2,\Co)$) gives the following fundamental operators.
\begin{definition}[\protect{\cite[Definition 13]{ABB:symop:2014CQGra..31m5015A}}]
Let the differential operators  
   $ \sDiv_{k,l}:\mathcal{S}_{k,l}\rightarrow
\mathcal{S}_{k-1,l-1}$, $\sCurl_{k,l}:\mathcal{S}_{k,l}\rightarrow
\mathcal{S}_{k+1,l-1}$,
$\sCurlDagger_{k,l}:\mathcal{S}_{k,l}\rightarrow
\mathcal{S}_{k-1,l+1}$, and
$\sTwist_{k,l}:\mathcal{S}_{k,l}\rightarrow \mathcal{S}_{k+1,l+1}$
be defined by
\begin{align*}
(\sDiv_{k,l}\varphi)_{A_1\dots A_{k-1}}{}^{A_1'\dots A_{l-1}'}={}&
\nabla^{BB'}\varphi_{A_1\dots A_{k-1}B}{}^{A_1'\dots A_{l-1}'}{}_{B'},\\
(\sCurl_{k,l}\varphi)_{A_1\dots A_{k+1}}{}^{A_1'\dots A_{l-1}'}={}&
\nabla_{(A_1}{}^{B'}\varphi_{A_2\dots A_{k+1})}{}^{A_1'\dots A_{l-1}'}{}_{B'},\\
(\sCurlDagger_{k,l}\varphi)_{A_1\dots A_{k-1}}{}^{A_1'\dots A_{l+1}'}={}&
\nabla^{B(A_1'}\varphi_{A_1\dots A_{k-1}B}{}^{A_2'\dots A_{l+1}')},\\
(\sTwist_{k,l}\varphi)_{A_1\dots A_{k+1}}{}^{A_1'\dots A_{l+1}'}={}&
\nabla_{(A_1}{}^{(A_1'}\varphi_{A_2\dots A_{k+1})}{}^{A_2'\dots A_{l+1}')}.
\end{align*}
The operators are called respectively the divergence, curl, curl-dagger, and twistor operators. 
\end{definition}
With respect to complex conjugation, the operators $\sDiv, \sTwist$ satisfy $\overline{\sDiv_{k,l}} = \sDiv_{l,k}$, $\overline{\sTwist_{k,l}} = \sTwist_{l,k}$, while $\overline{\sCurl_{k,l}} = \sCurlDagger_{l,k}$, $\overline{\sCurlDagger_{k,l}} = \sCurl_{l,k}$. In the following, we shall use the fundamental operators and their properties freely.
Any covariant expression in spinors and their covariant derivatives can be written in terms of the fundamental operators using the following Lemma. 
\begin{lemma}[\protect{\cite[Lemma 15]{ABB:symop:2014CQGra..31m5015A}}]
For any $\varphi_{A_1\dots A_k}{}^{A_{1}'\dots A_{l}'}\in \mathcal{S}_{k,l}$, we have the irreducible decomposition
\begin{align*}
\nabla_{A_1}{}^{A_1'}\varphi{}_{A_2\dots A_{k+1}}{}^{A_2'\dots A_{l+1}'}={}&
(\sTwist_{k,l}\varphi){}_{A_1\dots A_{k+1}}{}^{A_1'\dots A_{l+1}'}\nonumber\\
&-\tfrac{l}{l+1}\bar\epsilon^{A_1'(A_2'}(\sCurl_{k,l}\varphi){}_{A_1\dots A_{k+1}}{}^{A_3'\dots A_{l+1}')}\nonumber\\
&-\tfrac{k}{k+1}\epsilon_{A_1(A_2}(\sCurlDagger_{k,l}\varphi){}_{A_3\dots A_{k+1})}{}^{A_1'\dots A_{l+1}'}\nonumber\\
&+\tfrac{kl}{(k+1)(l+1)}\epsilon_{A_1(A_2}\bar\epsilon^{A_1'(A_2'}(\sDiv_{k,l}\varphi){}_{A_3\dots A_{k+1})}{}^{A_3'\dots A_{l+1}')}.
\end{align*}
\end{lemma}
For example, the Maxwell equation and the Killing spinor equations take the form 
$$
(\sCurlDagger_{2,0} \phi)_{AA'} = 0, 
$$
and 
\begin{equation*}
(\sTwist_{2,0} \kappa)_{ABCA'} = 0 
\end{equation*} 
respectively, in terms of the fundamental operators. 

In the computations below we shall need some commutator relations satisfied by the fundamental operators, see \cite[Lemma 18]{ABB:symop:2014CQGra..31m5015A}. The following lemma gives the commutators which are relevant here. 

\begin{lemma}
\label{lemma:commutators}
Let $\varphi_{AB} \in \mathcal{S}_{2,0}$. The operators $\sDiv$, $\sCurl$, $\sCurlDagger$ and $\sTwist$ satisfies the following commutator relations
\begin{subequations}
\begin{align}
(\sDiv_{1,1} \sCurlDagger_{2,0} \varphi)={}&0,\label{eq:DivCurlDagger}\\
(\sCurl_{3,1} \sTwist_{2,0} \varphi)_{ABCD}={}&2 \Psi_{(ABC}{}^{F}\varphi_{D)F},\label{eq:CurlTwist}\\
(\sCurlDagger_{3,1} \sTwist_{2,0} \varphi)_{ABA'B'}={}& 
\tfrac{2}{3} (\sTwist_{1,1} \sCurlDagger_{2,0} \varphi)_{ABA'B'}
+ 2 \Phi_{(A}{}^{C}{}_{|A'B'|}\varphi_{B)C},\label{eq:CurlDaggerTwist}\\
(\sDiv_{3,1} \sTwist_{2,0} \varphi)_{AB}={}&
-  \tfrac{4}{3} (\sCurl_{1,1} \sCurlDagger_{2,0} \varphi)_{AB}
-8 \Lambda \varphi_{AB}
 + 2 \Psi_{ABCD} \varphi^{CD} .\label{eq:DivTwistCurlCurlDagger}
\end{align}
\end{subequations}
\end{lemma}

Directly from the Killing spinor equation and the commutators \eqref{eq:DivCurlDagger} and \eqref{eq:DivTwistCurlCurlDagger} we get
\begin{subequations}\label{eq:Twistxi}
\begin{align}
(\sDiv_{1,1}\xi)=0,\\
(\sTwist_{1,1} \xi)_{ABA'B'}={}&-3 \Phi_{(A}{}^{C}{}_{|A'B'|}\kappa_{B)C}.
\end{align}
\end{subequations}
Hence, if the aligned matter condition is satisfied, $\xi^{AA'}$ is a Killing vector. 

Given a conformal Killing vector $\xi^{AA'}$, we define a conformally weighted Lie derivative acting on a symmetric valence $(2s,0)$ spinor field by \cite[Definition 17]{ABB:symop:2014CQGra..31m5015A}
\begin{align}
\hat{\mathcal{L}}_{\xi}\varphi_{A_1\dots A_{2s}}={}&\xi^{BB'} \nabla_{BB'}\varphi_{A_1\dots A_{2s}}+s \varphi_{B(A_2\dots A_{2s}} \nabla_{A_1)B'}\xi^{BB'} \nonumber \\
& + \tfrac{1-s}{4} \varphi_{A_1\dots A_{2s}} \nabla^{CC'}\xi_{CC'}. \label{eq:Liespin-def} 
\end{align}

We shall now prove an auxiliary result on the
derivatives of $\eta_{AA'}$, which will allow us to prove our main
result.
\begin{lemma} \label{lem:etaeqs} 
Let $\kappa_{AB}\in\mathcal{S}_{2,0}$ satisfy the Killing spinor
equation \eqref{eq:KillingSpinordef} and the aligned matter condition
\eqref{eq:alignedmatter}, and let $\xi_{AA'}$ be given by
\eqref{eq:xidef}. If $\phi_{AB}\in\mathcal{S}_{2,0}$ satisfies
the Maxwell equation \eqref{eq:Maxwell} and $\eta_{AA'}$ is given by
\eqref{eq:etadef}, then 
\begin{subequations} \label{eq:etafacts} 
\begin{align}
(\sDiv_{1,1} \eta)={}&0,\label{diveta1}\\
(\sCurl_{1,1} \eta)_{AB}={}&\tfrac{2}{3} (\hat{\mathcal{L}}_{\xi}\phi)_{AB},\label{curleta1b2}\\
(\sCurlDagger_{1,1} \eta)_{A'B'}={}&0,  \label{curleta2}\\
\eta_{AA'} \xi^{AA'}={}&\kappa^{AB} (\hat{\mathcal{L}}_{\xi}\phi)_{AB}. \label{eq:etaLphi} 
\end{align}
\end{subequations}
\end{lemma} 
\begin{proof}
Using the definition of the Lie derivative, the Maxwell equation and that $\xi^{AA'}$ is a Killing vector we get
\begin{align}\label{eq:Liexiphieq1}
(\hat{\mathcal{L}}_{\xi}\phi)_{AB}={}&\phi_{(A}{}^{C}(\sCurl_{1,1} \xi)_{B)C}
 + \xi^{CA'} (\sTwist_{2,0} \phi)_{ABCA'}.
\end{align}
The equation \eqref{diveta1} follows directly from the commutator relation \eqref{eq:DivCurlDagger}. 
Also using the commutators \eqref{eq:DivTwistCurlCurlDagger}, \eqref{eq:CurlTwist} and the Killing spinor equation, we get
\begin{align}
(\sCurl_{1,1} \xi)_{AB}={}&(\sCurl_{1,1} \sCurlDagger_{2,0} \kappa)_{AB}
=-6 \Lambda \kappa_{AB}
 + \tfrac{3}{2} \Psi_{ABCD} \kappa^{CD},\label{eq:curlxi}\\
0={}&\tfrac{1}{2} (\sCurl_{3,1} \sTwist_{2,0} \kappa)_{ABCD}
=\Psi_{(ABC}{}^{F}\kappa_{D)F}.\label{eq:intcondkappa}
\end{align}
Performing an irreducible decomposition of the contraction $\Psi_{ABCF} \kappa_{D}{}^{F}$, and using \eqref{eq:curlxi} and \eqref{eq:intcondkappa} we get
\begin{align}\label{eq:PsikappaTr1}
\Psi_{ABCF} \kappa_{D}{}^{F}={}&3 \Lambda \epsilon_{(A|D|}\kappa_{BC)}
 + \tfrac{1}{2} \epsilon_{(A|D|}(\sCurl_{1,1} \xi)_{BC)}.
\end{align}
By using the definition of $\Theta_{AB}$, the Leibniz rule, applying irreducible decompositions, and making use of the Killing spinor equation, the fact that $\xi_{AA'}$ is Killing, and the Maxwell equation, we find 
\begin{align}
(\sCurl_{1,1} \eta)_{AB}={}&(\sCurl_{1,1} \sCurlDagger_{2,0} \Theta)_{AB}\nonumber\\
={}&\kappa^{CD} (\sCurl_{3,1} \sTwist_{2,0} \phi)_{ABCD}
 + \tfrac{1}{2} \kappa_{(A}{}^{C}(\sDiv_{3,1} \sTwist_{2,0} \phi)_{B)C}
  \nonumber \\
& + \tfrac{4}{3} \phi_{(A}{}^{C}(\sCurl_{1,1} \xi)_{B)C} + \tfrac{2}{3} \xi^{CA'} (\sTwist_{2,0} \phi)_{ABCA'} \nonumber .
\end{align} 
Applying the commutator relations \eqref{eq:DivTwistCurlCurlDagger} and \eqref{eq:CurlTwist} and making use of \eqref{eq:PsikappaTr1} now gives 
\begin{align*}  
(\sCurl_{1,1} \eta)_{AB} ={}&\tfrac{2}{3} \phi_{(A}{}^{C}(\sCurl_{1,1} \xi)_{B)C}
 + \tfrac{2}{3} \xi^{CA'} (\sTwist_{2,0} \phi)_{ABCA'}\nonumber\\
={}&\tfrac{2}{3} (\hat{\mathcal{L}}_{\xi}\phi)_{AB} , \nonumber
\end{align*}
where \eqref{eq:Liexiphieq1} was used in the last step. 

Proceeding in a fashion similar to the above, using the definitions of $\eta_{AA'}$ and $\Theta_{AB}$, the Leibniz rule, applying irreducible decompositions, and making use of the Killing spinor equation, the fact that $\xi_{AA'}$ is Killing, and the Maxwell equation, we find
\begin{align*}
(\sCurlDagger_{1,1} \eta)_{A'B'}={}&\kappa^{AB} (\sCurlDagger_{3,1} \sTwist_{2,0} \phi)_{ABA'B'} .
\end{align*}
The commutator relation \eqref{eq:CurlDaggerTwist} then gives
\begin{align*}
(\sCurlDagger_{1,1} \eta)_{A'B'}={}&-2 \Phi_{BCA'B'} \kappa^{AB} \phi_{A}{}^{C},
\end{align*}
and the aligned matter condition gives \eqref{curleta2}.

Finally, expanding the definition of $\eta_{AA'}$, and using the Killing spinor equation and the Maxwell equation yields
\begin{align}
\kappa^{BC} (\sTwist_{2,0} \phi)_{ABCA'}={}&\eta_{AA'}
 + \tfrac{4}{3} \xi^{B}{}_{A'} \phi_{AB}.\label{eq:kappatwsitphitoeta} 
\end{align}
Contracting \eqref{eq:Liexiphieq1} with $\kappa_{AB}$ and using \eqref{eq:kappatwsitphitoeta}, \eqref{eq:curlxi}, and \eqref{eq:intcondkappa} gives \eqref{eq:etaLphi}.
\end{proof}

The proof of the main theorem is now a matter of straightforward verification.  
\begin{proof}[Proof of Theorem~\ref{thm:MainResultTensorVersion}]
From the Leibniz rule, we first find
\begin{align*}
\nabla^{BB'}V_{ABA'B'}={}&\tfrac{1}{2} \bar{\eta}_{A'B} \nabla^{BB'}\eta_{AB'}
 + \tfrac{1}{2} \bar{\eta}_{B'A} \nabla^{BB'}\eta_{BA'} \\
& + \tfrac{1}{2} \eta_{AB'} \nabla^{BB'}\bar{\eta}_{A'B} 
 + \tfrac{1}{2} \eta_{BA'} \nabla^{BB'}\bar{\eta}_{B'A} \\
&  + \tfrac{1}{3} \Theta_{AB} \nabla^{BB'}(\hat{\mathcal{L}}_{\bar{\xi}}\bar{\phi})_{A'B'}
 + \tfrac{1}{3} \overline{\Theta}_{A'B'} \nabla^{BB'}(\hat{\mathcal{L}}_{\xi}\phi)_{AB}\\
&
 + \tfrac{1}{3} \nabla^{BB'}\Theta_{AB} (\hat{\mathcal{L}}_{\bar{\xi}}\bar{\phi})_{A'B'}
  + \tfrac{1}{3} \nabla^{BB'}\overline{\Theta}_{A'B'}
  (\hat{\mathcal{L}}_{\xi}\phi)_{AB}.
\end{align*}
This can be simplified by first observing that
$\hat{\mathcal{L}}_{\xi}$ is a symmetry operator taking
solutions of the Maxwell equation to solutions of the Maxwell
equation, so $(\sCurlDagger_{2,0} \hat{\mathcal{L}}_{\xi}\phi)_{AB}=0$
and similarly for the complex conjugate. It can be further simplified
by substituting the definition $\nabla^B{}_{A'}\Theta_{AB}=\eta_{AA'}$, cf. \eqref{eq:etadef}, 
to eliminate the derivative of $\Theta_{AB}$ terms. This yields
\begin{align*}
\nabla^{BB'}V_{ABA'B'}
={}&- \tfrac{1}{2} \bar{\eta}_{A'}{}^{B} (\sCurl_{1,1} \eta)_{AB}
 -  \tfrac{1}{2} \eta^{B}{}_{A'} (\sCurl_{1,1} \bar{\eta})_{AB}
 -  \tfrac{1}{2} \bar{\eta}^{B'}{}_{A} (\sCurlDagger_{1,1} \eta)_{A'B'}\nonumber\\
& -  \tfrac{1}{2} \eta_{A}{}^{B'} (\sCurlDagger_{1,1} \bar{\eta})_{A'B'}
 + \tfrac{1}{2} \bar{\eta}_{A'A} (\sDiv_{1,1} \eta)
 + \tfrac{1}{2} \eta_{AA'} (\sDiv_{1,1} \bar{\eta}) \nonumber\\
&+ \tfrac{1}{3} \eta_{A}{}^{B'} (\hat{\mathcal{L}}_{\bar\xi}\bar{\phi})_{A'B'}
  + \tfrac{1}{3} \bar{\eta}_{A'}{}^{B} (\hat{\mathcal{L}}_{\xi}\phi)_{AB}.
\end{align*}
The terms involving $(\sCurlDagger_{1,1}\eta)_{A'B'}$ and
$(\sCurl_{1,1}\bar{\eta})_{AB}$ are zero by equation \eqref{curleta2}. Those
involving $(\sDiv_{1,1}\eta)$ and $(\sDiv_{1,1}\bar\eta)$ are zero by
equation \eqref{diveta1}. Finally by equation \eqref{curleta1b2}, the terms
involving $(\sCurl_{1,1}\eta)_{AB}$ and $(\sCurlDagger_{1,1}\bar{\eta})_{A'B'}$ cancel
with those involving $(\hat{\mathcal{L}}_{\xi}\phi)_{AB}$ and
$(\hat{\mathcal{L}}_{\bar\xi}\bar{\phi})_{A'B'}$ respectively. This completes the
result. 
\end{proof}

\section{Further remarks on the Kerr spacetime} \label{sec:kerr-remarks}
The stationary, asymptotically flat, vacuum Kerr spacetimes, and more generally the electro-vacuum Kerr-Newman spacetimes, have algebraic type $\{2, 2\}$, i.e. the Weyl spinor $\Psi_{ABCD}$ has two distinct, repeated,  principal spinors $o_A, \iota_A$ which are unique up to a rescaling. The dyad $o_A, \iota_A$ is normalized by $o_A \iota^A = 1$. 
For the following discussion, recall that given a spin dyad $o_A, \iota_A$, one defines for a symmetric spinor $\varpi_{A_1 \cdots A_k}$  scalars $\varpi_i$ by contracting $i$ times with $\iota^A$ and $k-i$ times with $o^A$.  
This yields Weyl scalars $\Psi_i$, $i=0, \cdots 4$ and Maxwell scalars $\phi_i$, $i=0,1,2$. In a spacetime of type $\{2,2\}$ with principal dyad $o_A, \iota_A$, it holds that $\Psi_{ABCD} = 6 \Psi_2 o_{(A} o_B \iota_C \iota_{D)}$, and in this case it follows from \eqref{eq:intcondkappa} that any valence $(2,0)$ Killing spinor must be of the form 
\begin{equation}
\kappa_{AB} = \zeta o_{(A} \iota_{B)},\label{eq:kappadyad}
\end{equation}
for some scalar $\zeta$. 

If $(t,r,\theta,\phi)$
are Boyer-Lindquist coordinates,  then
the Coulomb field, i.e. the unique static, regular Maxwell test field, on the Kerr-Newman spacetime takes the form 
\begin{equation*}
\phi_{AB} = 
\frac{1}{(r-ia\cos\theta)^2} o_{(A} \iota_{B)}
\end{equation*} 
up to a rescaling by a constant. In particular the extreme components $\phi_0, \phi_2$ are zero. The background Maxwell field in the electro-vacuum Kerr-Newman spacetime is a constant multiple of 
this Coulomb field. 

The Killing spinor $\kappa_{AB}$ is 
\begin{equation}\label{eq:kappaBL}
\kappa_{AB}=\tfrac{2}{3}(r-ia\cos\theta)o_{(A}\iota_{B)}
\end{equation} 
which is therefore proportional to the background Maxwell field in the Kerr-Newman spacetime. Hence, by the Einstein equation, $\Phi_{ABA'B'}$ is proportional to $\kappa_{AB} \bar{\kappa}_{A'B'}$. It follows that the aligned matter condition holds in the Kerr-Newman spacetime.

The normalisation in equation \eqref{eq:kappaBL} is chosen 
so that $\xi^a=(\partial_t)^a$, where $\xi_a$ is given by \eqref{eq:xidef}. In particular $\xi_a$ is real, which exhibits the fact that the Kerr-Newman family admits a Killing-Yano tensor, as remarked above. 
In particular, we see that the tensor $V_{ab}$ given by \eqref{eq:Vdef} is conserved. More generally, any vacuum type $\{2,2\}$ spacetime admits a Killing spinor of valence $(2,0)$, of the form \eqref{eq:kappadyad} 
with $\zeta$ proportional to $\Psi_2^{-1/3}$. This shows that Theorem \ref{thm:MainResultTensorVersion} applies in the class of vacuum type $\{2,2\}$ metrics. 

\subsection{The Teukolsky equations and $V_{ab}$} 
The Maxwell equations on a Kerr black hole imply the $s=1$ 
Teukolsky equations for the extreme
scalars, $\phi_{0}$ and $\phi_2$. This system has many properties in common with the $s=2$
Teukolsky equations which arise from linearising the
Einstein equations.
Despite the fact that the Teukolsky equations have been known
for more than 40 years, and have been the subject of much study, no boundedness or decay estimates are known for the $s\ne 0$ Teukolsky equations, other than the mode stability result of Whiting \cite{whiting:1989}.

In terms of the Maxwell scalars $\phi_i$, the Newman-Penrose scalars for $\Theta_{AB}$ satisfy
\begin{align*}
\Theta_0&=-2\kappa_1\phi_0,&
\Theta_1&=0,&
\Theta_2&=2\kappa_1\phi_2 .
\end{align*}
Thus, only the extreme components of $\phi_{AB}$ appear in $\Theta_{AB}$, and
hence in $\eta_{AA'}$. Equation \eqref{eq:etaLphi} can be used to express 
$(\hat{\mathcal{L}}_{\xi}\phi)_{AB}$ in terms of $\eta_{AA'}$ and $(\hat{\mathcal{L}}_{\xi}\Theta)_{AB}$, from which it follows
that $V_{ab}$ can be written solely in terms of the extreme
components of $\phi_{AB}$. This has two important consequences. Firstly, in the Kerr-Newman
spacetime the extreme components of the Coulomb solutions vanish, 
and hence the conserved tensor 
$V_{ab}$ naturally excludes
non-radiating solutions of the Maxwell equation. Secondly, since it is defined in terms of the extreme Maxwell scalars alone, $V_{ab}$ can
be thought of as an ``energy-momentum tensor'' for the $s=1$ combined Teukolsky/Teukolsky-Starobinsky system, which corresponds to equations \eqref{curleta1b2}-\eqref{curleta2}, see \cite{ABB:jubilee} for more details.

\subsection*{Acknowledgements} Part of the work on this paper was carried out during visits to the Erwin Schr\"odinger Institute, Vienna, and the Mathematical Sciences Research Institute, Berkeley. We are grateful to these institutions for hospitality and support. The authors thank Steffen Aksteiner for many enlightening discussions. 

\newcommand{\arxivref}[1]{\href{http://www.arxiv.org/abs/#1}{{arXiv.org:#1}}}
\newcommand{\mnras}{Monthly Notices of the Royal Astronomical Society}
\newcommand{\prd}{Phys. Rev. D}


\begin{thebibliography}{1}

\bibitem{anco:pohjanpelto:2003:ProcRSoc:MR1997098}
S.~C. Anco and J.~Pohjanpelto.
\newblock Conserved currents of massless fields of spin {$s\geq\frac 12$}.
\newblock {\em R. Soc. Lond. Proc. Ser. A Math. Phys. Eng. Sci.},
  459(2033):1215--1239, 2003.

\bibitem{ABB:currents}
L.~{Andersson}, T.~{B{\"a}ckdahl}, and P.~{Blue}.
\newblock {Conserved currents}.
\newblock In preparation.

\bibitem{ABB:jubilee}
L.~{Andersson}, T.~{B{\"a}ckdahl}, and P.~{Blue}.
\newblock {Spin geometry and conservation laws in the Kerr spacetime}.
\newblock In preparation.

\bibitem{ABB:symop:2014CQGra..31m5015A}
L.~{Andersson}, T.~{B{\"a}ckdahl}, and P.~{Blue}.
\newblock {Second order symmetry operators}.
\newblock {\em Classical and Quantum Gravity}, 31(13):135015, July 2014.
\newblock \arxivref{1402.6252}.

\bibitem{Bae11a}
T.~{B\"{a}ckdahl}.
\newblock Sym{M}anipulator, 2011-2014.
\newblock
  \href{http://www.xact.es/SymManipulator}{http://www.xact.es/SymManipulator}.

\bibitem{bergqvist:etal:2003CQGra..20.2663B}
G.~{Bergqvist}, I.~{Eriksson}, and J.~M.~M. {Senovilla}.
\newblock {New electromagnetic conservation laws}.
\newblock {\em Classical and Quantum Gravity}, 20:2663--2668, July 2003.
\newblock \arxivref{gr-qc/0303036}.

\bibitem{xAct}
J.~M. Mart\'{\i}n-Garc\'{\i}a.
\newblock x{A}ct: {E}fficient tensor computer algebra for {M}athematica,
  2002-2014.
\newblock \href{http://www.xact.es}{http://www.xact.es}.

\bibitem{Penrose:1986fk}
R.~Penrose and W.~Rindler.
\newblock {\em {Spinors and Space-time I {\&} II}}.
\newblock Cambridge Monographs on Mathematical Physics. Cambridge University
  Press, Cambridge, 1986.

\bibitem{whiting:1989}
B.~F. {Whiting}.
\newblock {Mode stability of the Kerr black hole}.
\newblock {\em Journal of Mathematical Physics}, 30:1301--1305, June 1989.

\end{thebibliography}
\end{document}